\documentclass[a4paper,11pt]{article}
\usepackage[margin=1.2in]{geometry}

\renewenvironment{abstract}
{\small
\vspace{-1em}
\begin{center}
\bfseries \abstractname\vspace{-.5em}\vspace{0pt}
\end{center}
\list{}{
\setlength{\leftmargin}{0.6in}%
\setlength{\rightmargin}{\leftmargin}%
}%
\item\relax}
{\endlist}

\usepackage[utf8]{inputenc}
\usepackage[T1]{fontenc}
\usepackage[textwidth=2.5cm, color=yellow]{todonotes}
\usepackage{amsmath, amsthm, amssymb}
\usepackage{graphicx}
\usepackage[colorlinks=true, citecolor=blue, linkcolor=blue]{hyperref}
\usepackage[linesnumbered, vlined, ruled]{algorithm2e}
\usepackage[noadjust]{cite}

\usepackage{needspace}

\makeatletter
\newcommand{\otherlabel}[2]{\protected@edef\@currentlabel{#2}\label{#1}}
\makeatother

\usepackage{aliascnt}
\newtheorem{theorem}{Theorem}[section]

\newaliascnt{lemma}{theorem}
\newtheorem{lemma}[lemma]{Lemma}
\aliascntresetthe{lemma}

\newaliascnt{proposition}{theorem}

\aliascntresetthe{proposition}

\newaliascnt{definition}{theorem}
\newtheorem{definition}[definition]{Definition}
\aliascntresetthe{definition}

\newaliascnt{corollary}{theorem}
\newtheorem{corollary}[corollary]{Corollary}
\aliascntresetthe{corollary}

\newaliascnt{conjecture}{theorem}

\aliascntresetthe{conjecture}

\theoremstyle{remark}
\newaliascnt{claim}{theorem}

\aliascntresetthe{claim}

\newaliascnt{observation}{theorem}

\aliascntresetthe{observation}

\newaliascnt{example}{theorem}

\aliascntresetthe{example}

\newaliascnt{remark}{theorem}

\aliascntresetthe{remark}

\def\cqedsymbol{\ifmmode$\lrcorner$\else{\unskip\nobreak\hfil
\penalty50\hskip1em\null\nobreak\hfil$\lrcorner$
\parfillskip=0pt\finalhyphendemerits=0\endgraf}\fi}

\usepackage{tabularx}
\usepackage{environ}
\makeatletter
\newcommand{\problemtitle}[1]{\gdef\@problemtitle{#1}}
\newcommand{\probleminput}[1]{\gdef\@probleminput{#1}}
\newcommand{\problemquestion}[1]{\gdef\@problemquestion{#1}}
\NewEnviron{decproblem}{
\problemtitle{}\probleminput{}\problemquestion{}
\BODY
\par\addvspace{.5\baselineskip}
\noindent
\begin{tabularx}{\textwidth}{@{\hspace{\parindent}} l X c}
\multicolumn{2}{@{\hspace{\parindent}}l}{\@problemtitle} \\
\textbf{Input:} & \@probleminput \\
\textbf{Question:} & \@problemquestion
\end{tabularx}
\par\addvspace{.5\baselineskip}}
\NewEnviron{genproblem}{
\problemtitle{}\probleminput{}\problemquestion{}
\BODY
\par\addvspace{.5\baselineskip}
\noindent
\begin{tabularx}{\textwidth}{@{\hspace{\parindent}} l X c}
\multicolumn{2}{@{\hspace{\parindent}}l}{\@problemtitle} \\
\textbf{Input:} & \@probleminput \\
\textbf{Output:~~} & \@problemquestion
\end{tabularx}
\par\addvspace{.5\baselineskip}}

\def\TransEnum{\textsc{Trans-Enum}}
\def\HypergraphDual{\textsc{Hypergraph Dualization}}
\def\Dual{\textsc{Dual}}
\def\DualEnum{\textsc{DualEnum}}

\def\B{\mathcal{B}} 
\def\C{\mathcal{C}} 
\def\G{\mathcal{G}} 
\def\H{\mathcal{H}} 
\def\L{\mathcal{L}} 
\def\Sig{(X,\Sigma)} 

\newcommand{\intv}[2]{[#2]}

\DeclareMathOperator{\Min}{Min}
\DeclareMathOperator{\Max}{Max}

\DeclareMathOperator{\poly}{poly}

\DeclareMathOperator{\spex}{\mathbf{mingen}}
\DeclareMathOperator{\ex}{\mathbf{ex}}
\DeclareMathOperator{\dex}{\mathbf{min}}

\usepackage{old-arrows}

\definecolor{belize}{RGB}{0, 143, 216}
\definecolor{teal}{RGB}{0, 180, 166}

\begin{document}

\title{Dualization in lattices given\\ by implicational bases\thanks{A preliminary version of this article appeared in the proceedings of the 15\textsuperscript{th} International Conference on Formal Concept Analysis, ICFCA 2019~\cite{defrain2019dualization}.	 The authors have been supported by the ANR project GraphEn ANR-15-CE40-0009.}}

\author{
    Oscar Defrain\thanks{LIMOS, Université Clermont Auvergne, France.}
    \addtocounter{footnote}{-1} 
    \and
    Lhouari Nourine\footnotemark
}

\maketitle

\begin{abstract}
	It was recently proved that the dualization in lattices given by implicational bases is impossible in output-polynomial time unless {\sf P$=$NP}.
	In this paper, we~show that this result holds even when the premises in the implicational base are of size at most two.
	Then we show using hypergraph dualization that the problem can be solved in output quasi-polynomial time whenever the implicational base has bounded independent-width, defined as the size of a maximum set of implications having independent conclusions.
	Lattices that share this property include distributive lattices coded by the ideals of an interval order, when both the independent-width and the size of the premises equal one.

	\vskip5pt\noindent{}{\bf Keywords:} lattice dualization, transversals enumeration, implicational base, distributive lattice, interval order.
\end{abstract}


\section{Introduction}\label{sec:introduction}

The dualization of a monotone Boolean function is ubiquitous in many areas of computer science including database theory, logic and artificial intelligence~\cite{eiter1995identifying,gunopulos1997data,eiter2003new,nourine2012extending}.
When defined on Boolean lattices, the problem is equivalent to the enumeration of the minimal transversals of a hypergraph, arguably one of the most important open problems in algorithmic enumeration by now~\cite{eiter1995identifying,eiter2008computational}.
In this case, the best known algorithm is due to Fredman and Khachiyan and runs in output quasi-polynomial time~\cite{fredman1996complexity}.
An enumeration algorithm is said to be running in {\em output-polynomial} time if its running time is bounded by a polynomial in the combined size of the input and the output~\cite{johnson1988generating}.
When generalized to any lattice, it was recently proved by Babin and Kuznetsov in~\cite{babin2017dualization} that the dualization is impossible in output-polynomial time unless {\sf P$=$NP}.
This result holds under two different settings, when the lattice is given by an {\em implicational base}, or by the ordered set of its {\em irreducible elements} (by its {\em context} in FCA terminology).
The first representation consists of implications of the form $A \rightarrow b$ over a ground set $X$, i.e., $A,\{b\}\subseteq X$, that express the fact that if an element contains $A$ in the lattice, then it must contain~$b$;
$A$ is called the {\em premise}, and $b$ the {\em conclusion} of the implication.
A set of implications is called {\em independent} if every of its implications has a conclusion that cannot be obtained from the premises of other implications.
The {\em dimension} of an implicational base is the size of one of its largest premise.
In~the first setting when the lattice is given by an implicational base, the observation in~\cite{babin2017dualization} is based on a result of Kavvadias et al.~\cite{kavvadias2000generating} on the intractability of enumerating the maximal models of a Horn expression.
The constructed implicational base, however, has an implication with a premise of unbounded size, and the tractability status of the dualization remained open in the case of implicational bases of bounded dimension.
In~this paper, we address this problem with the following result.

\begin{theorem}\label{thm:main-npc}
	The dualization in lattices given by implicational bases is impossible in output-polynomial time unless {\sf P$=$NP}, even for implicational bases of dimension two.
\end{theorem}

In~the case of premises of size one, the problem remains open.
The lattice in that situation is distributive~\cite{davey2002introduction}.
The best known algorithm is due to Babin and Kuznetsov and runs in output sub-exponential time~\cite{babin2017dualization}.
We show using hypergraph dualization that it can be solved in output quasi-polynomial time whenever the implicational base has bounded {\em independent-width}, defined as the size of a maximum independent set of implications.
Our result holds in fact in the more general context of implicational bases having unbounded dimension.
See Theorem~\ref{thm:bounded-main}.
The approach is similar to the one in~\cite{nourine2014dualization} as we show that the problem can be reduced to hypergraph dualization in that case, which allows us to use the algorithm of Fredman and Khachiyan.
Lattices that share this property include distributive lattices coded by the ideals of an interval order, when both the independent-width and the dimension of the implicational base equal one.

The rest of the paper is organized as follows.
In~Section~\ref{sec:prelim} we introduce necessary concepts and definitions.
Theorems~\ref{thm:main-npc} and~\ref{thm:bounded-main} are respectively proved in Sections~\ref{sec:npc} and~\ref{sec:bounded}.
We~conclude with future research directions in Section~\ref{sec:concl}.

\section{Preliminaries}\label{sec:prelim}

All objects considered in this paper are finite.
For a set $X$ we denote by $2^X$ the set of all subsets of $X$.
For an integer $n\in \mathbb{N}$ we denote by $\intv{1}{n}$ the set $\{1,\dots,n\}$.
We shall note $f\in \poly(n)$ if $f: \mathbb{N}\rightarrow \mathbb{N}$ is a function and $f\in O(n^c)$ for some fixed constant $c\in \mathbb{N}$.

A {\em hypergraph} $\H$ over a ground set $X$ is a subset $\H$ of $2^X$.
Elements of $\H$ are called {\em hyperedges}, and elements of $X$ are called {\em vertices}.
A {\em transversal} of $\H$ is a subset $T\subseteq X$ of vertices that intersects every hyperedge $E\in \H$.
It is called {\em minimal} if it is minimal by inclusion.
The set of all minimal transversals of $\H$ is denoted by $Tr(\H)$.
Note that $Tr(\H)$ also defines a hypergraph.
The problem of deciding whether $\G=Tr(\H)$ given two hypergraphs $\H,\G\subseteq 2^X$ is known as \HypergraphDual{}.
The problem of computing $Tr(\H)$ given $\H\subseteq 2^X$ is denoted by \TransEnum{}.
It is well known that there is a polynomial-time algorithm for \HypergraphDual{} if and only if there is an output-polynomial time algorithm for \TransEnum{}~\cite{bioch1995complexity,eiter2008computational}.
To date, the best known algorithm for these problems is due to Fredman and Khachiyan~\cite{fredman1996complexity} and runs in $N^{o(\log N)}$ time where $N=|\H|+|\G|$.
The existence of a polynomial-time algorithm solving \HypergraphDual{} is now open for more than 35 years~\cite{eiter1995identifying,eiter2003new,eiter2008computational}.
We refer the reader to~\cite{eiter2008computational} for a survey on hypergraph dualization.

A {\em partial order} on a set $X$ (or {\em poset}) is a binary relation $\leq$ on $X$ which is reflexive, anti-symmetric and transitive, denoted by $P=(X,\leq)$.
Two elements $x$ and $y$ of $P$ are said to be {\em comparable} if $x \leq y$ or $y \leq x$, otherwise they are said to be {\em incomparable}.
We~note $x<y$ if $x\leq y$ and $x\neq y$.
If an element $u$ of $P$ is such that both $x\leq u$ and $y\leq u$ then $u$ is called {\em upper bound} of $x$ and $y$; it is called {\em least upper bound} of $x$ and $y$ if moreover $u\leq v$ for every upper bound $v$ of $x$ and~$y$.
Note that two elements of a poset may or may not have a least upper bound.
The least upper bound (also known as {\em supremum} or {\em join}) of $x$ and $y$, if it exists, is denoted by $x\vee y$. 
The greatest lower bound (also known as {\em infimum} or {\em meet}) of $x$ and $y$, if it exists, is denoted by $x\wedge y$ and is defined dually. 
A~subset of a poset in which every two elements are comparable is called a {\em chain}. 
A~subset of a poset in which no two distinct elements are comparable is called an {\em antichain}.
A poset is an {\em interval order} if it corresponds to an ordered collection of intervals on the real line such that $[x_1,x_2] < [x_3,x_4]$ if and only if $x_2<x_3$.
The {\sf 2+2} poset is the union of two disjoint $2$-elements chains.
It is well known that interval orders are {\sf 2+2}-free, that is, they do not induce the {\sf 2+2} poset as a suborder~\cite{fishburn1970intransitive}.
A~set $I\subseteq X$ is called {\em ideal} of $P$ if $x\in I$ and $y\leq x$ imply $y\in I$. 
If $x\in I$ and $x\leq y$ imply $y\in I$, then $I$ is called {\em filter} of $P$.
Note that the complementary of an ideal is a filter, and vice versa.
For~every $x\in P$ we associate the {\em principal ideal of $x$} (or simply {\em ideal of $x$}), denoted by $\downarrow\,x$, and defined by $\downarrow x=\{y\in X \mid y\leq x\}$. 
The {\em principal filter of} $x\in X$ is the dual $\uparrow x=\{y\in X \mid x\leq y\}$.
If $S$ is a subset of $X$, we respectively denote by $\downarrow S$ and $\uparrow S$ the sets defined by $\downarrow S=\bigcup_{x\in S} \downarrow x$ and $\uparrow S=\bigcup_{x\in S} \uparrow x$, and denote by $\Min(S)$ and $\Max(S)$ the sets of minimal and maximal elements of $S$ with respect to~$\leq$ in $P$.
The following notion is central in this paper.

\begin{definition}
	Let $P=(X,\leq)$ be a poset and $B^+$, $B^-$ be two antichains of $P$.
	We~say that $B^+$ and $B^-$ are dual in $P$ if $\downarrow\!B^+\,\cup \uparrow\!B^-=X$ and $\downarrow\!B^+\,\cap \uparrow\!B^-=\emptyset$.
\end{definition}

In other words, $B^+$ and $B^-$ are dual in $P$ if one of ${B^+=\Max\{x \mid x\not\in\,\uparrow\!B^-\}}$ or ${B^-=\Min\{x \mid x\not\in\,\downarrow\!B^+\}}$ holds.
Hence the problem of deciding whether two antichains $B^+$ and $B^-$ of $P$ are dual can be solved in polynomial time in the size of $P$.
The task becomes difficult when the poset is not fully given, but only an implicit coding---of possibly logarithmic size in the size of $P$---is given:
this is usually the case when considering dualization problems in lattices.

\begin{figure}
	\center
	\includegraphics[scale=1.1]{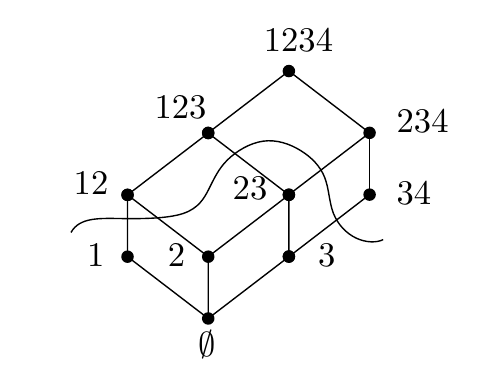}
	\caption{The lattice $\L(\Sigma)$ of closed sets of the implicational base $\Sigma=\{13\rightarrow 2,\ 4\rightarrow 3\}$ on ground set $X=\{1,2,3,4\}$, and the border (curved line) formed by the two dual antichains $\B^+=\{\{1\},\{2,3\}\}$ and $\B^-=\{\{1,2\}, \{3,4\}\}$ of $\L(\Sigma)$. For better readability, closed sets and premises are denoted without braces, i.e., $123$ stands for $\{1,2,3\}$.}\label{fig:lattice-acyc}
\end{figure}

A {\em lattice} is a poset in which every two elements have a least upper bound and a greatest lower bound~\cite{davey2002introduction,gratzer2011lattice}.
It is called {\em distributive} if for any three elements $x,y,z$ of the lattice,
\[
	x \wedge (y \vee z)=(x \wedge y) \vee (x \wedge  z).
\]
An implicational base $(X,\Sigma)$ is a set $\Sigma$ of implications of the form $A \rightarrow B$ where $A\subseteq X$ and $B\subseteq X$; see~\cite{wild2017joy,bertet2018lattices}.
In this paper we only consider implicational bases in their equivalent {\em unit} form where $|B|=1$ for every implication, and denote by $A\rightarrow b$ such implications, where $B=\{b\}$.
The {\em size} of $\Sigma$ is the number of implications in~$\Sigma$.
It is denoted by $|\Sigma|$.
The {\em dimension} of $\Sigma$ is the size of a largest premise in~$\Sigma$.
A set $C\subseteq X$ is {\em closed} in $\Sigma$ if for every implication $A\rightarrow b$ of $\Sigma$, at least one of $b\in C$ or $A\not\subseteq C$ holds.
To~$\Sigma$ we associate the {\em closure operator} $\phi$ which maps every $C\subseteq X$ to the smallest closed set of $\Sigma$ containing $C$, and that we denote by $\phi(C)$.
Then, we note $\C_\Sigma$ the set of all closed sets of $\Sigma$.
It~is well known that every lattice can be represented as the set of all closed sets of an implicational base, ordered by inclusion.
To~$\Sigma$ we associate $\L(\Sigma)=(\C_\Sigma,\subseteq)$ such a lattice.
We note that an antichain of $\L(\Sigma)$ is a set $\B\subseteq \C_\Sigma$ such that $B_1\not\subseteq B_2$ for any two $B_1,B_2\in \B$.
An example of a lattice of closed sets of an implicational base is given in Figure~\ref{fig:lattice-acyc}.
If~$\Sigma$ is empty, then $\L(\Sigma)=(2^X,\subseteq)$ and the lattice is called {\em Boolean}.
If~$\Sigma$ only has premises of size one, then the lattice is distributive and this is in fact a characterization~\cite{davey2002introduction}.
Furthermore in that case, the implicational base can be seen as a poset $P=(X,\leq)$ where $x\leq y$ if and only if $y\rightarrow x$, and $\phi(S)=\downarrow_P S$ for all $S\subseteq X$.
We call {\em underlying poset} of $\Sigma$ this poset.
Note that in general $\L(\Sigma)$ may be of exponential size in the size of $(X,\Sigma)$: it is in particular the case when the implicational base is empty.

In this paper, we are concerned with the following decision problem and one of its two generation versions.

\begin{decproblem}
	\problemtitle{Dualization in Lattices Given by Implicational Bases (\Dual{})}
	\probleminput{An implicational base $(X,\Sigma)$ and two antichains $\B^+,\B^-$ of $\L(\Sigma)$.}
	\problemquestion{Are $\B^+$ and $\B^-$ dual in $\L(\Sigma)$?}
\end{decproblem}

\begin{genproblem}
	\problemtitle{Generation version of \Dual{} (\textsc{DualEnum})}
	\probleminput{An implicational base $(X,\Sigma)$ and an antichain $\B^+$ of $\L(\Sigma)$.}
	\problemquestion{The dual antichain $\B^-$ of $\B^+$ in $\L(\Sigma)$.}
\end{genproblem}

A positive instance of \Dual{} is given in Figure~\ref{fig:lattice-acyc}.
Observe that the lattice $\L(\Sigma)$ is not given in any of the two problems defined above.
Only $(X,\Sigma)$ is given, which is a crucial point.
Recently in~\cite{babin2017dualization} it was shown that \Dual{} is co{\sf NP}-complete, hence that \DualEnum{} cannot be solved in output-polynomial time unless {\sf P$=$NP}.
When the implicational base is empty---when the lattice is Boolean---the problem is equivalent to \HypergraphDual{}.
Then it admits an algorithm running in $N^{o(\log N)}$ time where $N=|\B^+|+|\B^-|$ using the algorithm of Fredman and Khachiyan.
In the case of premises of size one---when the lattice is distributive---the best known algorithm is due to Babin and Kuznetsov~\cite{babin2017dualization} and runs in sub-exponential time $2^{O(n^{0,67} \log^3 N)}$ where $N=|\B^+|+|\B^-|$ and $n=|X|$.
Quasi-polynomial time algorithms are known for subclasses of distributive lattices, including products of chains \cite{elbassioni2009algorithms}.

We conclude the preliminaries with notions of width that we later consider in this paper.
Let $(X,\Sigma)$ be an implicational base and $\phi$ be its associated closure operator.
A~set $T\subseteq X$ is {\em independent} w.r.t.~$\phi$ if $x\not\in \phi(T\setminus \{x\})$ for any $x\in T$.
Given two sets $T,I\subseteq X$ we say that $T$ is a {\em covering set} of $I$ if $I\subseteq \phi(T)$, and that it is a {\em generating set} of $I$ if in addition $T\subseteq I$. 
It is called {\em minimal} if $I\not\subseteq \phi(T\setminus \{x\})$ for any $x\in T$.
Clearly, every minimal covering set of $I$ is independent, and a generating set of $I$ is minimal if and only if it is independent.
We point out that these notions only rely on $\phi$ and not on the implications in $\Sigma$.
To every $I\subseteq X$ we associate the set $\spex(I)\subseteq 2^I$ of minimal generating sets of $I$.
Note that several such subsets exist in general. 
We distinguish a particular one that we denote by $\ex(I)$ and that is obtained from $T=I$ by the following procedure:

\begin{center}
\begin{tabular}{l}
	{\bf while there exists} $x\in T$ {\bf such that} $I\subseteq \phi(T\setminus \{x\})$ {\bf do} $T\leftarrow T\setminus \{x\}$\\
	{\bf return} $T$ {\bf as} $\ex(I)$
\end{tabular}
\end{center}

\noindent In order for such a procedure to be deterministic we chose $x$ of smallest index in $T$ at each step.
A subset of implications in $\Sigma$ is called {\em independent} if every of its implications has a conclusion that cannot be obtained from the premises of other implications.
In other words, a set of $k$ implications $\{{A_1\rightarrow b_1},\,\dots\,,{A_k\rightarrow b_k}\}\subseteq\Sigma$ is {\em independent} if, for any $i\in\intv{1}{k}$, $b_i\not\in \phi(A_1\cup \dots \cup A_{i-1} \cup A_{i+1} \cup \dots \cup A_k)$.
We call {\em independent-width} of $(X,\Sigma)$ the size of a maximum independent set of implications in $\Sigma$.
Note that if $\Sigma$ is of dimension and independent-width one, then it has no implications $a\rightarrow b$ and $c\rightarrow d$ such that $d\not\in\phi(a)$ and $b\not\in\phi(c)$.
In particular, the underlying poset of $\Sigma$ is {\sf 2+2}-free in that case, i.e., it is an interval order.

\section{Implicational bases of dimension two}\label{sec:npc}

We show that it is co{\sf NP}-complete to decide whether two antichains of a lattice given by an implicational base of dimension two are dual.
The reduction is based on the one of Kavvadias et al.~in~\cite{kavvadias2000generating}, except that we manage to hide the Horn clause of unbounded size in one of the two antichains.

\begin{theorem}\label{thm:npc}
	\Dual{} is co{\sf NP}-complete for implicational bases of dimension two.
\end{theorem}

\begin{proof}
	Membership in co{\sf NP} follows from the fact that checking whether $\downarrow \B^+ \cap {\uparrow \B^-}\neq\emptyset$, or whether a given set $F\subseteq X$, closed in $\Sigma$, is such that both $F\not\in\,\downarrow\!\B^+$ and $F\not\in\,\uparrow\!\B^-$ can be done in polynomial time in the sizes of $\Sig$, $\B^+$ and $\B^-\!$; such a set $F$ constitutes a certificate for a `no' answer.

	We show completeness by reducing \textsc{One-in-Three 3Sat}, restricted to positive literals, to the complement of \Dual{}.
	This restricted case of \textsc{One-in-Three 3Sat} remains {\sf NP}-complete~\cite{kavvadias2000generating,garey2002computers}.
	In this problem, one is given a $n$-variable, $m$-clause positive Boolean formula
	\[
		\phi(x_1,\dots,x_n)=\bigwedge_{j=1}^m C_j = \bigwedge_{j=1}^m (c_{j,1} \vee c_{j,2} \vee c_{j,3})
	\]
	where $x_1,\dots,x_n$ and $C_1,\dots,C_m$ respectively denote the variables and the clauses of $\phi$, and where every variable appears in at least one clause ($c_{j,i}$ denotes the variable that appears in clause $j$ at position $i$).
	Then the task is of deciding whether there exists an assignment of the variables such that every clause contains exactly one variable to one.
	We call {\em one-in-three truth assignment} such an assignment.
	We construct an instance of \Dual{} as follows.
	Let $X=\{x_1,\dots,x_n,y_1,\dots,y_m,z\}$ be the ground set made of one element $x$ per variable of~$\phi$, one element $y$ per clause of~$\phi$, and an additional special element $z$.
	Let~$\Sigma$~be the implicational base defined by 
	\otherlabel{imp:1}{(1)}%
	\otherlabel{imp:2}{(2)}%
	\otherlabel{imp:3}{(3)}%
	\otherlabel{imp:4}{(4)}%
	\otherlabel{imp:5}{(5)}%
	\otherlabel{imp:6}{(6)}%
	\otherlabel{imp:7}{(7)}%
	\begin{equation*}
		\!\!\!\!\!\!\!\!\!\Sigma = 
		\left \{
		\begin{array}{rlr}
			~c_{j,1}c_{j,2}&\rightarrow~ z  &~ \text{(1)}~\\
			~c_{j,1}c_{j,3}&\rightarrow~ z  &~ \text{(2)}~\\
			~c_{j,2}c_{j,3}&\rightarrow~ z  &~ \text{(3)}~\\
			~zc_{j,1}&\rightarrow~ y_j   &~ \text{(4)}~\\
			~zc_{j,2}&\rightarrow~ y_j   &~ \text{(5)}~\\
			~zc_{j,3}&\rightarrow~ y_j   &~ \text{(6)}~\\
			~y_j&\rightarrow~ z &~ \text{(7)}~
		\end{array}
		\middle\vert~~ j\in\intv{1}{m}~\right\}.
	\end{equation*}
	Then we put
	\begin{gather*}
		\B^+=\{B_j=X\setminus \{y_j,c_{j,1},c_{j,2},c_{j,3}\}~ \mid~ j\in\intv{1}{m}\},\\ 
		\B^-=\{F=\{y_1,\dots,y_m,z\}\}.
	\end{gather*}
	Clearly, $(X,\Sigma)$, $\B^+$ and $\B^-$ are constructed in polynomial time in the size of $\phi$.
	Moreover, every $B_j\in \B^+$ is closed in $\Sigma$ (observe that no literal in $\{c_{j,1},c_{j,2},c_{j,3}\}$ is the conclusion of an implication of $\Sigma$, and that $y_j$ cannot be implied without any literal in $\{c_{j,1},c_{j,2},c_{j,3}\}$).
	As $B_j$ is the only set of $\B^+$ containing $y_j$ for every $j\in\intv{1}{m}$, no two sets in $\B^+$ are inclusion-wise comparable.
	Hence $\B^+$ is an antichain of $\L(\Sigma)$.
	Also, $\B^-$ is an antichain of $\L(\Sigma)$ as it is a singleton and its unique element $F$ is closed in~$\Sigma$.
	At last, $\Sigma$ is of dimension two.
	Are $\B^+$ and $\B^-$ dual in $\L(\Sigma)$?
	We show that the answer is `no' if and only if there is a one-in-three truth assignment of~$\phi$.

	We prove the first implication.
	Let us assume that $\B^+$ and $\B^-$ are not dual in $\L(\Sigma)$.
	Since $\downarrow \B^+\cap \uparrow \B^-=\emptyset$, there must be some closed set $F'\subseteq X$ such that both $F'\not\in\,\downarrow\!\B^+$ and $F'\not\in\,\uparrow\!\B^-\!$.
	We consider an inclusion-wise minimal such set $F'$.
	Since $F\setminus\{z\}$ is not closed in $\Sigma$, and $F\setminus \{y_j\}\subseteq B_j$ for every $j\in \intv{1}{m}$, we conclude that $F'\not\subseteq F$.
	Then $F'\cap \{x_1,\dots,x_n\}\neq\emptyset$.
	Let $x\in F'\cap \{x_1,\dots,x_n\}$.
	We show that $z\not\in F'$ by contradiction.
	Suppose that $z\in F'$.
	Then by Implications~\ref{imp:4} to~\ref{imp:6}, $y_j\in F'$ for all $j\in \intv{1}{m}$ such that $x\in C_j$.
	Hence for every clause $C_j$ containing~$x$, we have that $|F'\cap \{y_j,c_{j,1},c_{j,2},c_{j,3}\}|\geq 2$.
	Hence $F'\setminus \{x\}\not\subseteq B_j$ for any $j\in \intv{1}{m}$. 
	Since $F'\setminus \{x\}$ is closed, this contradicts the fact that $F'$ is chosen minimal such that $F'\not\in\,\downarrow\!\B^+\!$.
	Hence $F'$ does not contain $z$.
	Clearly $F'\cap\{y_1,\dots,y_m\}=\emptyset$ as otherwise by Implication~\ref{imp:7}, $F'$ would contain~$z$.
	As $F'\not\subseteq B_j$ for any $j\in \intv{1}{m}$, $|F'\cap C_j|\geq 1$ for every such $j$.
	Furthermore $|F'\cap C_j|\leq 1$ for every $j\in \intv{1}{m}$ as otherwise by Implications~\ref{imp:1} to~\ref{imp:3} $F'$ would contain~$z$.
	Consequently $F'$ is a one-in-three truth assignment of $\phi$, concluding the first implication.

	We prove the other implication.
	Let $T$ be a one-in-three truth assignment of $\phi$. 
	As $T\subseteq \{x_1,\dots,x_n\}$ and $|T\cap C_j|=1$ for all $j\in \intv{1}{m}$, $T$ is closed in $\Sigma$. 
	Furthermore it is not a subset of any $B_j\in \B^+\!$. 
	Since at last $T\not\supseteq F$, we obtain that both $T\not\in\,\downarrow\!\B^+$ and $T\not\in\,\uparrow\!\B^-\!$. 
	Consequently $\B^+$ and $\B^-$ are not dual in $\L(\Sigma)$, concluding the proof.
\end{proof}

As a consequence, there is no algorithm solving \DualEnum{} in output-polynomial time unless {\sf P$=$NP}, even in the case of implicational bases of dimension two.
This proves Theorem~\ref{thm:main-npc}.

\section{Implicational bases of bounded independent-width}\label{sec:bounded}

We show using hypergaph dualization that the dualization in lattices given by implicational bases can be achieved in output quasi-polynomial time whenever the implicational base has bounded independent-width.

In what follows, let $(X,\Sigma,\B^+)$ be an instance of \DualEnum{}.
Let $\B^-$ be the dual antichain of $\B^+$ that we wish to compute, and $\H$ be the \emph{complementary hypergraph} of $\B^+$ on ground set $X$ defined by
\[
	\H=\{X\setminus B \mid B\in \B^+\}.
\]
It is well known that $Tr(\H)=\B^-$ whenever $\L(\Sigma)$ is Boolean, that is when $\Sigma$ is empty.
We will show how $\B^-$ can be computed from $Tr(\H)$ in the general case.

\needspace{0.5in}
\begin{lemma}\label{lem:T-closure-1}
	To every transversal $T$ of $\H$ corresponds some $I\in \B^-$ such that $I\subseteq \phi(T)$.
	This is in particular the case for every minimal transversal of $\H$.
\end{lemma}

\begin{proof}
	Let $T$ be a transversal of $\H$. 
	As $T\cap E\neq\emptyset$ for all $E\in \H$, $T$ satisfies $T\not\subseteq B$ for any $B\in \B^+\!$.
	As $T\subseteq \phi(T)$ this is also the case of $\phi(T)$.
	Hence $\phi(T)\not\in \downarrow \B^+\!$.
	Now if there is no $I\in \B^-$ such that $I\subseteq \phi(T)$ then $\phi(T)\not\in \uparrow \B^-\!$, contradicting the duality of $\B^+$ and $\B^-$ in $\L(\Sigma)$.
	We conclude that one such $I$ must exist.
	The last remark follows by inclusion.
\end{proof}

\needspace{0.5in}
\begin{lemma}\label{lem:T-ex}
	If $T$ is a transversal of $\H$, then every set $T^*$ in $\spex(T)$ is.
	In particular, every minimal transversal of $\H$ is independent w.r.t.~$\phi$.
\end{lemma}

\begin{proof}
	We proceed by contradiction.
	Let $T$ be a transversal of $\H$ and $T^*\in\spex(T)$.
	Suppose that $T^*$ is not a transversal.
	Then $T^*\subseteq B$ for some $B\in \B^+$.
	As $B$ is closed, $\phi(T^*)\subseteq \phi(B)=B$.
	Since $T^*$ is a generating set of $T$, $T\subseteq \phi(T^*)$.
	Hence $T\subseteq B$ and thus $T$ is not a transversal of $\H$, a contradiction.
	Consequently every $T^*\in\spex(T)$ is a transversal of $\H$.
	In particular, every minimal transversal $T$ of $\H$ is independent w.r.t.~$\phi$, as otherwise it can be reduced into an arbitrary independent generating set of $T$ which is smaller, contradicting the minimality of $T$.
\end{proof}

\begin{lemma}\label{lem:T-closure-2}
	To every $I\in \B^-$ corresponds $T\in Tr(\H)$ such that $T=\ex(I)$.
\end{lemma}

\begin{proof}
	Let $I\in \B^-\!$.
	Since $I\not\subseteq B$ for any $B\in \B^+\!$, $I$ is a transversal of $\H$.
	Let $T=\ex(I)$.
	By Lemma~\ref{lem:T-ex} as $\ex(I)\in\spex(I)$, $T$ is a transversal of $\H$ and since $I$ is closed, $I=\phi(T)$.
	We show that $T$ is minimal.
	Let $x\in T$ and $I'=\phi(T\setminus \{x\})$.
	As $T$ is a minimal generating set of $I$, $I'\subset I$.
	By minimality of $I$ it must be that $I'\subseteq B$ for some $B\in \B^+\!$.
	Consequently $I'$ does not intersect the hyperedge $E=X\setminus B$ for such a $B$.
	As $T\setminus \{x\}\subseteq I'$, $T\setminus \{x\}$ is not a transversal of $\H$.
	We conclude that $T\in Tr(\H)$.
\end{proof}

\begin{algorithm}
	\SetAlgoLined

	$\H\leftarrow\{X\setminus B \mid B\in \B^+\}$\;

	\For{{\em\bf every} $T\in Tr(\H)$\label{line:forall}}
	{   
		$I\leftarrow \phi(T)$\;

		\If{$I\in \B^-$ {\bf and} $T=\ex(I)$}
		{

			{\bf output} $I$\; 
		}
	}

	\caption{An algorithm enumerating the dual antichain $\B^-$ of $\B^+$ in $\L(\Sigma)$ given an implicational base $(X,\Sigma)$ of closure operator $\phi$ and an antichain $\B^+$ of $\L(\Sigma)$.}\label{algo:main}
\end{algorithm}

A consequence of Lemma~\ref{lem:T-closure-2} is that one can enumerate $\B^-$ from $Tr(\H)$ by checking for every $T\in Tr(\H)$ whether its closure $I=\phi(T)$ belongs to $\B^-\!$, whether $T=\ex(I)$, and discarding the solution if not. 
Computing $I=\phi(T)$ can be done in $O(|X|\cdot |\Sigma|)$ time.
Testing whether $I$ belongs to $\B^-$ can be done in $O(|X|^2 \cdot (|\Sigma|+|\B^+|))$ time by checking for every $x\in I$ whether $I\setminus \{x\}$ is not closed, or whether $I\setminus \{x\}\subseteq B$ for some $B\in \B^+$ otherwise.
This holds since $I$ is a transversal and if $I\setminus \{x\}\subseteq B$, then $I\setminus \{x\}$ is not a transversal, which establishes that $I\in B^-$.
As for the computation of $\ex(I)$ it can be done in $O(|X|^2\cdot |\Sigma|)$ time following the definition in Section~\ref{sec:prelim}.
Henceforth, enumerating $\B^-$ can be done in total time 
\[
	M^{o(\log M)}+|Tr(\H)|\cdot O(|X|^2 \cdot (|\Sigma|+|\B^+|))
\]
where $M=|\H|+|Tr(\H)|$, by constructing $\H$ in $O(|X| \cdot |\B^+|)$ time, using the algorithm in~\cite{fredman1996complexity} for the enumeration of $Tr(\H)$ in time $M^{o(\log M)}$, and discarding at most $|Tr(\H)|$ solutions with a cost of $O(|X|^2 \cdot (|\Sigma|+|\B^+|))$ per solution.
Repetitions are avoided by discarding $T$ whenever $T\neq\ex(I)$.
This procedure is given in Algorithm~\ref{algo:main}.
Its correctness follows from Lemmas~\ref{lem:T-closure-1} and~\ref{lem:T-closure-2}.
The limitation of such a procedure is that the size of $Tr(\H)$ may be exponentially larger than that of $X$, $\Sigma$, $\B^+$ and $\B^-\!$, hence that the described algorithm may run in output-exponential time.
An example of one such instance is given in Figure~\ref{fig:expinstance-dimension}.
However, we will show that it is not the case whenever the implicational base has bounded independent-width.

Our argument relies on the following observation.

\begin{figure}[t]
	\center
	\includegraphics[scale=1.1]{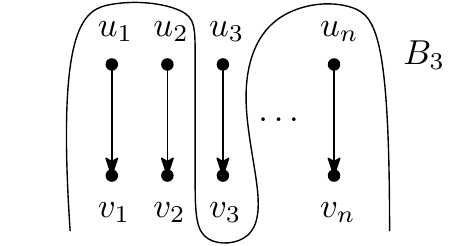}
	\caption{An implicational base $(X,\Sigma)$ on ground set $X=\{u_1,v_1,\dots,u_n,v_n\}$ where $\Sigma=\{u_i\rightarrow v_i \mid i\in \intv{1}{n}\}$.
	By taking $\B^+=\{X\setminus \{u_i,v_i\} \mid i\in \intv{1}{n}\}$, we get $\H=\{\{u_i,v_i\}\mid i\in \intv{1}{n}\}$, $\B^-=\{\{v_1,\dots,v_n\}\}$ and $Tr(\H)=\{\{z_1,\dots,z_n\}\mid (z_1,\dots,z_n)\in \{u_1,v_1\}\times\dots\times\{u_n,v_n\}\}$.}\label{fig:expinstance-dimension}
\end{figure}

\begin{lemma}\label{lem:transversal-indmingen}
	Let $I\in \B^-$ and $T$ be a minimal transversal of $\H$ such that $I\subseteq \phi(T)$.
	Then $T$ is a minimal covering set of $I$.
\end{lemma}

\begin{proof}
	First recall that by Lemma~\ref{lem:T-ex}, $T$ is independent.
	It may intersect $I$.
	Let $x\in T$.
	By minimality of $T$, $T\setminus \{x\}$ is not a transversal.
	By Lemma~\ref{lem:T-ex}, neither is $\phi(T\setminus \{x\})$ as otherwise since $T$ is independent then $T\setminus \{x\} \in \spex(\phi(T\setminus \{x\}))$ is a transversal, which contradicts the hypothesis that $T$ is minimal.
	Since $I$ is a transversal of $\H$ we have that $I\not\subseteq \phi(T\setminus \{x\})$ for any $x\in T$ and the lemma follows.
\end{proof}

In the following given two subsets $T,I\subseteq X$ such that $T$ is an independent covering set of $I$, we note $\dex(\Sigma,T,I)$ an arbitrary minimal subset of implications of $\Sigma$ having their premise included in $T$ as a subset and that are needed in $\Sigma$ in order to derive $I$ from $T$.
In other words, $\dex(\Sigma,T,I)$ is obtained from the implications of $\Sigma$ having their premise in $T$ by greedily removing an implication of $\Sigma$ having its premise in $T$ while the inclusion $I\subseteq \phi(T)$ holds.
Observe that in consequence no implication in $\dex(\Sigma,T,I)$ has a conclusion that is obtained by closure of the other premises in $\dex(\Sigma,T,I)$, i.e., $\dex(\Sigma,T,I)$ is an independent set of implications of $\Sigma$.

We now express a bound on the number of minimal covering sets a set admits depending on the number of implications in $\Sigma$, and its independent-width.
By Lemma~\ref{lem:transversal-indmingen}, this yields a bound on the number of minimal transversals of $\H$ depending on the sizes of $\B^-\!$, $\Sigma$, and the independent-width of $\Sigma$.

\begin{theorem}\label{thm:bound-indmingen}
	Let $I$ be a subset of $X$.
	Then the number of minimal covering sets of $I$ is bounded by $|\Sigma|^k$ where $k$ is the independent-width of $\Sigma$.
\end{theorem}

\begin{proof}
	Let $I\subseteq X$ and $T\subseteq X$ be a minimal covering set of~$I$.
	Consider the implications in $\dex(\Sigma,T,I)$.
	As $\dex(\Sigma,T,I)$ is an independent set of implications, ${|\dex(\Sigma,T,I)|\leq k}$.
	Observe in addition that every $x\in T\setminus I$ belongs to at least one premise of an implication in $\dex(\Sigma,T,I)$, as otherwise $T$ is not a minimal covering set of $I$.
	Furthermore by definition, every $x$ that belongs to the premise of an implication in $\dex(\Sigma,T,I)$ is in $T$.
	Hence, every such $x$ is either in $I$ or in $T\setminus I$.
	We conclude that $T\setminus I=\bigcup\{A \mid A\rightarrow b \in \dex(\Sigma,T,I)\}\setminus I$.
	Now, observe that $T$ is uniquely characterized by $T\setminus I$ as $T$ is independent: the elements of $T\cap I$ are exactly those of $I\setminus\phi(T\setminus I)$.
	Since $T\setminus I$ is obtained by union of at most $k$ implications in $\Sigma$, the number of minimal covering sets of~$I$ is bounded by
	\[
		\sum_{i=1}^{k} \binom{|\Sigma|}{i}
	\]
	hence by $|\Sigma|^k$, as desired.
\end{proof}

A corollary of Lemma~\ref{lem:transversal-indmingen} and Theorem~\ref{thm:bound-indmingen} is the following, observing that every solution $I\in \B^-$ admits at most $|\Sigma|^k$ minimal covering sets, hence that at most $|\Sigma|^k$ minimal transversals of $\H$ have their closure containing $I$.

\begin{corollary}\label{cor:bound-tr}
	If $\Sigma$ is of independent-width $k$ then $|Tr(\H)|\leq |\Sigma|^k\cdot |\B^-|$.
\end{corollary}

As a consequence, the size of $Tr(\H)$ is bounded by a polynomial in $|X|+|\Sigma|+|\B^+|+|\B^-|$ whenever the implicational base is of bounded independent-width.
Hence under such a condition, it is still reasonable to test each of the minimal transversals generated by Algorithm~\ref{algo:main} even though many may not lead to a solution of $\B^-\!$.
We conclude with the following theorem.

\begin{theorem}\label{thm:bounded-main}
	There is an algorithm that, for every integer $k$, given an implicational base $(X,\Sigma)$ such that $\Sigma$ is of independent-width $k$, and an antichain $\B^+$ of $\L(\Sigma)$, enumerates the dual antichain $\B^-$ of $\B^+$ in $\L(\Sigma)$ in output quasi-polynomial time $N^{o(\log N)}$ where $N=|X|+|\Sigma|+|\B^+|+|\B^-|$.
\end{theorem}

\begin{proof}
	Let $k$ be an integer and $(X,\Sigma)$ be an implicational base of independent-width $k$.
	Let $\B^+$ be an antichain of $\L(\Sigma)$, and $\H=\{X\setminus B \mid B\in \B^+\}$ be the complementary hypergraph of $\B^+\!$.
	Let $\B^-$ be the dual antichain of $\B^+$ in $\L(\Sigma)$ that we wish to compute.
	By Corollary~\ref{cor:bound-tr}, the size of $Tr(\H)$ is bounded by $|X|^k\cdot|\B^-|$.
	Let $M=|\H|+|Tr(\H)|$ and $N=|X|+|\Sigma|+|\B^+|+|\B^-|$.
	Since $|\H|\leq |\B^+|$, there exists a constant $c\in \mathbb{N}$ depending in $k$ such that $M\leq N^c$.
	As a consequence, using the algorithm of Fredman and Khachiyan, the running time of Algorithm~\ref{algo:main} on instance $(X,\Sigma,\B^+)$ is bounded by 
	\[
		M^{o(\log M)}+|Tr(\H)|\cdot O(|X|^2 \cdot (|\B^+| + |\Sigma|))
	\]
	hence by
	\[
		N^{c\,\cdot\,o(\log N^c)}+\poly(N)=N^{o(\log N)}.
	\]
\end{proof}

As a corollary, there is a quasi-polynomial time algorithm solving \Dual{} in lattices given by implicational bases of bounded independent-width.
In particular if the dimension and the independent-width of $\Sigma$ equal one, Theorem~\ref{thm:bounded-main} yields an output quasi-polynomial time algorithm solving \DualEnum{} in distributive lattices coded by the ideals of an interval order (see~Section~\ref{sec:prelim}).

\section{Conclusion and future work}\label{sec:concl}

In this paper, we showed that the dualization in lattices given by implicational bases is impossible in output-polynomial time unless {\sf P$=$NP}, even when the premises in the implicational base are of dimension two.
Then, we~showed using hypergraph dualization that the problem admits an output quasi-polynomial time algorithm whenever the implicational base has bounded independent-width.
Lattices that share this property include distributive lattices coded by the ideals of an interval order when both the independent-width and the dimension of the implicational base equal one.

We state open problems for future research.
To an implicational base $(X,\Sigma)$ we associate its {\em implication-graph} $G(\Sigma)$ as the directed graph on vertex set $X$ and where there is an arc from $x$ to $y$ if there exists $A\rightarrow b\in \Sigma$ such that $x\in A$ and $y=b$.
An implicational base $(X,\Sigma)$ is called {\em acyclic} if $G(\Sigma)$ has no directed cycle.
Acyclic implicational bases have been widely studied in the literature~\cite{hammer1995quasi,boros2009subclass,wild2017joy}.
Observe that the negative result of Section~\ref{sec:npc} involves an implicational base which is cyclic.
Consequently, an important research direction concerns the dualization in lattices given by acyclic implicational bases.
Subclasses of interest include distributive lattices as we recall that the best known algorithm for the dualization in that case is output sub-exponential~\cite{babin2017dualization}.
Superclasses of interest that are not covered by Theorem~\ref{thm:npc} include convex geometries.

\section*{Acknowledgement}

The first author would like to thank Simon Vilmin for extensive discussions on the topic of this paper.

\bibliographystyle{alpha}
\bibliography{main}

\end{document}